\documentclass[11pt]{amsart}
\usepackage{fullpage}
\usepackage{amsthm,amsmath,amssymb}
\usepackage[foot]{amsaddr}

\usepackage{enumitem}
\usepackage{booktabs}
\usepackage{graphicx}
\usepackage{subcaption}
\usepackage{tikz}
\usepackage{natbib}

\usepackage[colorlinks=true,linkcolor=red,breaklinks=true,citecolor=blue,urlcolor=cyan]{hyperref}

\graphicspath{{./figs/}} 

\newtheorem{theorem}{Theorem}[section]
\newtheorem{proposition}[theorem]{Proposition}
\newtheorem{lemma}[theorem]{Lemma}
\newtheorem{definition}[theorem]{Definition}
\theoremstyle{remark}

\title{Observational Causality Testing}

\author{Brian Knaeble}
\address{Department of Computer Science, Utah Valley University, Orem, UT}
\email{bknaeble@uvu.edu}

\author{Braxton Osting}
\address{Department of Mathematics, University of Utah, Salt Lake City, UT}
\email{osting@math.utah.edu}

\author{Placede Tshiaba}
\address{Department of Mathematics, University of Utah, Salt Lake City, UT}
\email{placede.tshiaba@utah.edu}


\begin{document}
\maketitle

\begin{abstract}
In prior work we have introduced an asymptotic threshold of sufficient randomness for causal inference from observational data. In this paper we extend that prior work in three main ways. First, we show how to empirically estimate a lower bound for the randomness from measures of concordance transported from studies of monozygotic twins. Second, we generalize our methodology for application on a finite population and we introduce methods to implement finite population corrections. Third, we generalize our methodology in another direction by incorporating measured covariate data into the analysis. The first extension represents a proof of concept that observational causality testing is possible. The second and third extensions help to make observational causality testing more practical. As a theoretical and indirect consequence of the third extension we formulate and introduce a novel criterion for covariate selection. We demonstrate our proposed methodology for observational causality testing with numerous example applications.
\end{abstract}

\section{Introduction}
\label{intro}
A sensitivity analysis can support causal inference from observational data. Methods for sensitivity analysis introduce sensitivity parameters that can be computed from observed data. A classic approach involving propensity probabilities is described in \citet[Chapter 4]{Rose}. More recently, \citet{Oster} has utilized coefficients of determination for the purpose of sensitivity analysis. 
In related, prior work we have introduced an asymptotic threshold of sufficient randomness for causal inference \citep{Knaeble2023}. Here we expand ideas further, beyond sensitivity analysis, resulting in a methodology that we refer to as observational causality testing. This paper represents a proof of concept that observational causality testing is possible. 
\subsection{Methods}
Here we assume that the natural process giving rise to an exposure is individualistic and probabilistic \cite[p. 31]{Imbens2015}, and we assume also the Stable Unit Treatment Value Assumption (SUTVA) \cite[p. 10]{Imbens2015}. To make SUTVA more plausible we operate within a framework of stochastic counterfactuals \citep{Robins1989,VR,Kent2019}, wherein SUTVA is relaxed to mean that the parameter values of the stochastic, potential outcomes are stable. Each individual of the population has not just a propensity probability of treatment but also two prognosis probabilities of the outcome in the presence and absence of treatment, c.f. \citet{Hansen2008,Leacy2013,Aikens2022}. The process giving rise to the outcome is also assumed to be individualistic and probabilistic.  

Our methodology is general enough to be applicable whenever an event $E$ is thought to cause another event $D$. The event $E$ could be a treatment or an exposure, and the event $D$ could be a response, an outcome, or the occurrence of a disease. We indicate occurrence of $E$ with $e=1$ and the absence of $E$ with $e=0$. Likewise, we indicate the occurrence of $D$ with $d=1$ and the absence of $D$ with $d=0$. We denote individual propensity probability of $E$ with $\pi$, we denote individual prognosis probability of $D$ in the absence of $E$ with $r_{0}$, and we denote individual prognosis probability of $D$ in the presence of $E$ with $r_{1}$. For each individual indexed by $i$ we have $e_i\sim \textrm{Bernoulli}(\pi_i)$, and the stochastic, potential outcomes are $d_i(e_i=0)\sim \textrm{Bernoulli}(r_{0i})$ and $d_i(e_i=1)\sim \textrm{Bernoulli}(r_{1i})$. For each individual $i$ our generalized SUTVA means $(r_{0i},r_{1i})$ does not depend on $e_i$, and we assume also that $(d_i(e_i=0),d_i(e_i=1))$ is independent of $e_i$. Unless otherwise specified we assume an infinite population of individuals. We observe the population distribution of $(e,d)$ values, but we do not know the distribution of latent $(\pi,r_{0},r_{1})$ values. Our null hypothesis of no causality is
\begin{equation}
\label{null}
H_0\colon  r_{0}=r_{1} ~\textrm{almost surely}.
\end{equation}

%
We denote individual, expected, prognosis probability with $r:=\pi r_{1}+(1-\pi)r_{0}$. 
We denote population means with $\bar{\pi}$ and $\bar{r}$, and we denote population variances with $\sigma^2_\pi$ and $\sigma^2_r$. The randomness is defined as \begin{equation}\label{etadef}\eta=1-R_{\pi}R_r,\end{equation} where $R^2_\pi=\sigma^2_\pi/(\bar{\pi}(1-\bar{\pi}))$ and $R^2_r=\sigma^2_r/(\bar{r}(1-\bar{r}))$. Empirical evidence for randomness in data generating processes can be found within studies of monozygotic twins. For an event $\zeta\in\{E,D\}$ we write $BC_\zeta$ for the probandwise concordance, $PC_\zeta$ for the pairwise concordance, and $P$ for the proportion of individuals for whom the event $\zeta$ occurred. The following inequalities are derived from the results of Appendix \ref{concordances}: 
\begin{equation}\label{lb1}l_\eta(BC):=1-\sqrt{1-\frac{1-BC_E}{1-P(e=1)}}\sqrt{1-\frac{1-BC_D}{1-P(d=1)}}\leq \eta,\end{equation}
\begin{equation}\label{lb2}l_\eta(PC):=1-\sqrt{1-\frac{1-PC_E}{(1+PC_E)(1-P(e=1))}}\sqrt{1-\frac{1-PC_D}{(1+PC_D)(1-P(d=1))}}\leq \eta.\end{equation}
The inequalities tighten when measured concordances come from studies of monozygotic twins reared apart.

We consider theoretical, non-causal explanations of the observed data. Let $\mathcal P(S)$ be the space of distributions $\mu$ on the unit square $S:=\{(\pi,r):0<\pi,r<1\}$. A feasible, non-causal, explanation of the observed data is a distribution $\mu\in\mathcal P(S)$ satisfying
\begin{subequations}
\label{margconst}
\begin{align}
&\int (1-\pi)r_0 d\mu=P(e=0,d=1), \\
&\int \pi r_1 d\mu=P(e=1,d=1), \\ 
&\int (1-\pi)(1-r_0)d\mu=P(e=0,d=0), \textrm{and} \\
&\int \pi(1-r_1)d\mu=P(e=1,d=0).
\end{align}\end{subequations}
The threshold of sufficient randomness for causal inference is defined as 
\begin{equation}\label{Tdef}T:=\max_{\mathcal P(S)} \eta.\end{equation}

Write $p_{01}=P(e=0,d=1)$, $p_{11}=P(e=1,d=1)$, $p_{00}=P(e=0,d=0)$, and $p_{10}=P(e=1,d=0)$. Define the $\phi$ coefficient 
\[\phi=\frac{p_{11}p_{00}-p_{01}p_{10}}{\sqrt{(p_{11}+p_{10})(p_{00}+p_{10})(p_{01}+p_{11})(p_{01}+p_{00})}}.\]
In prior work we have proven the formula
\begin{equation}
\label{Tformula}
    T=1-|\phi|
\end{equation}
\citep[Theorem 1]{Knaeble2023}.
A warrant for causal inference is provided by $\eta>T$.

\subsection{Introductory Example Application}
Does cigarette smoking cause chronic obstructive pulmonary disease (COPD)? \citet{Terz} observed a positive association between smoking and COPD. In the table below the exposure, $E$, is smoking, and the outcome, $D$, is COPD.
\begin{table}[ht!]
\centering
\caption{A contingency table showing an observed association (Relative Risk, RR $=2.8$) between Cigarette Smoking and Chronic Obstructive Pulmonary Disease (COPD).}
\label{COPD}
\begin{tabular}{rcc}
\toprule
 & \multicolumn{2}{c}{Smoking}\\
\cmidrule{2-3}
COPD&No&Yes\\
\hline
Yes&318&1,631\\
No&4,679&7,538\\
\hline
\end{tabular}
\end{table}

Among monozygotic twins, the probandwise concordance of smoking has been estimated to be about $BC_E=67\%$ \citep{Kaprio1984} and the probandwise concordance of COPD has been estimated to be about $BC_D=20\%$ \citep{Ingebrigtsen2010}. From the table we compute $P(e=1)=0.65$, $P(d=1)=0.14$, and $\phi(e,d)=0.16$. With these values in (\ref{lb1}) and (\ref{Tformula}) we obtain $\eta\geq l_\eta=0.79<0.84=T$, which does not on its own provide a warrant for causal inference. Theoretically, in the absence of further assumptions or scientific knowledge, given our propensity-prognosis model of the data generating process, unmeasured covariates could provide a non-causal explanation for the data in Table \ref{COPD}; our empirical lower bound on the randomness was not sufficient. Later in Section \ref{success} we will meet an example application where the empirical lower bound on the randomness is sufficient and the observational data does provide a warrant for causal inference. 
\subsection{Outline}

In a prior paper we showed how to compute the threshold $T$ asymptotically \citep{Knaeble2023}. In this paper we make the following additional contributions. We
\begin{itemize}
\item estimate an empirical lower bound for $\eta$ from measures of concordance transported from studies of monozygotic twins,
\item relate the threshold $T$ to common measures of association in Section \ref{sens},
\item implement a finite population correction to conduct observational causality testing with $(1-\alpha)\times 100\%$ confidence in Section \ref{fpc},
\item incorporate measured covariate data for higher powered tests of causality in Section \ref{scov}, and
\item introduce a criterion for covariate selection in Section \ref{csel}.
\end{itemize}
A discussion occurs in Section \ref{disc}. Numerous example applications are provided throughout.
\section{The Threshold, $T$, as a Sensitivity Parameter}
\label{sens}
We proceed to describe formulas for computing $T$ from the observed prevalence of the exposure, the observed prevalence of the disease, and an observed measure of association. 
Write the standard (open) 3-simplex 
\begin{equation}
\label{simp}
\Delta := \left\{ (p_{01},p_{11},p_{00},p_{10})\in\mathbb{R}^4 \colon \ 
p_{01}, \ p_{11}, \ p_{00}, \ p_{10}>0; 
\ \  p_{01}+p_{11}+p_{00}+p_{10}=1 \right\}.
\end{equation}
Denote the marginal relative frequencies by $p_e=P(e=1)$ and $p_d=P(d=1)$. 
The set $\Delta$ can be reparametrized with $(p_e,p_d,\xi)$, where $\xi$ is a measure of association. We have already considered $\xi=\phi$. Next we consider common measures of association known as the risk difference (RD), 
\[RD=\frac{p_{11}}{(p_{11}+p_{10})}-\frac{p_{01}}{(p_{01}+p_{00})},\]
the relative risk (RR), 
\[RR=\frac{p_{11}}{(p_{11}+p_{10})}\frac{(p_{01}+p_{00})}{p_{01}},\]
and the odds ratio (OR),
\[OR=\frac{p_{11}}{p_{10}}\frac{p_{00}}{p_{01}}.\]

While (\ref{Tformula}) gives an explicit formula for $T$ in terms of the relative frequencies of a contingency table,  
the following proposition provides formulas for computing $T$ from observed $(p_e,p_d,\xi)$, where $\xi\in\{RD,RR,OR\}$. We anticipate that the following Proposition may prove useful during systematic reviews or meta analyses.
\begin{proposition}
\label{P}
Suppose $p_e,p_d\in(0,1)$. Define $\lambda=\sqrt{\frac{p_e(1-p_e)}{p_d(1-p_d)}}$. For $\xi \in \{RD,RR,OR\}$, 
\begin{equation}
\label{RDformula}
    T=
    \begin{cases}
 1-|RD|\lambda & \textrm{if } \xi=RD\\
  1- \left|\frac{p_d(RR-1)}{1+p_e(RR-1)} \right| \lambda & \textrm{if } \xi=RR\\
  1-\left|\frac{p_d(u-1)}{1+p_e(u-1)} \right| \lambda & \textrm{if } \xi=OR,\\
\end{cases}.
\end{equation}
where $u=\frac{-a+\sqrt{a^2-4 p_e m}}{2 p_e}$, $a = p_d\left(OR-1\right)+\left(1-p_e\right)-p_eOR$, and $m = (p_e-1)OR$.
\end{proposition}
\noindent Proposition \ref{P} is proven in Appendix \ref{A1}.

The middle row of (\ref{RDformula}) shows that $T$ is decreasing in $p_d$ given $RR$ and $p_e$. Studies of rare outcomes may require more randomness to warrant causal inference. From the data of Table \ref{COPD} we can compute $RR=2.8$ and $T=0.84$. The data in the following Table \ref{T1} were taken from \citet{Steg}. From the data of Table \ref{T1} we can compute $RR=5.8$ and $T=0.87$. Despite the higher relative risk value the threshold has risen. This is primarily due to the observed prevalence of diabetes being lower than the observed prevalence of COPD. Further insights of Proposition \ref{P} with illustrations and corollaries can be found in \citet[Section 3]{https://doi.org/10.48550/arxiv.2103.05692}.

\begin{table}[ht]
\centering
\caption{A contingency table showing an observed association (RR $=5.8$, $T=0.87$) between diabetes and stroke.}
\label{T1}
\begin{tabular}{rcc}
\toprule
 & \multicolumn{2}{c}{Diabetes}\\
\cmidrule{2-3}
Stroke&No&Yes\\
\hline
Yes&1,823&647\\
No&110,986&6,277\\
\hline
\end{tabular}
\end{table}
\section{Finite Population Analysis}
\label{fpc}
On a smaller finite population of size $n$ we should relax the equality constraints of (\ref{margconst}).
Let $\mathcal P_n(S)$ denote the space of distributions on $S=\{(\pi,r):0<\pi,r<1\}$, where each distribution $\mu_n\in \mathcal P_n(S)$ has $n$ point masses of equal weight indexed by $i$. Define the vector valued functional
\begin{equation}
\label{vfunctional}
f(\mu_n)=\left(\sum_{i=1}^n(1-\pi_i)r_{i}, \sum_{i=1}^n\pi_ir_{i}, \sum_{i=1}^n(1-\pi_i)(1-r_{i}), \sum_{i=1}^n\pi_i(1-r_{i})\right).
\end{equation}

Let $x_0=(np_{01}, np_{11}, np_{00}, np_{10})$ be our observed vector of frequency counts, and write $p_0=x_0/n$. The multinomial distribution $Mult(n,p_0)$ has expectation equal to $x_0$ and we may denote its covariance matrix with $\Sigma_m$. The matrix $\Sigma_m$ is a function of $x_0$. It is singular but it has a Moore-Penrose pseudo-inverse which we denote with $\Sigma_m^\dag$.

Given $\alpha$ with $0<\alpha<1$ we may replace $\mathcal P_n(S)$ with 
\begin{equation}
\label{finiteform}
\mathcal P_{n,1-\alpha}(S)=\{\mu_n\in \mathcal P_n(S): (x_0-f(\mu_n))^t \Sigma_{m}^{\dag}(x_0-f(\mu_n))<\chi^2_{3}(1-\alpha)\}.   
\end{equation}
Given $x_0$ the set $\mathcal P_{n,1-\alpha}(S)$ is an approximate and conservative $(1-\alpha)\times 100\%$ confidence set for $\mu_n$; see Appendix \ref{appmatrix} for details. We may define
\begin{equation}
    \label{finiteT}
T_n(1-\alpha)=\max_{\mathcal P_{n,1-\alpha}(S)}\eta. 
\end{equation} 
A warrant for causal inference is provided by $\eta>T_n(1-\alpha)$, with $(1-\alpha)\times 100\%$ confidence.

\subsection{Computing a Finite Population Correction}
\label{fpcs}
To compute $T_n(1-\alpha)$ in practice we resample from $x_0$ with replacement. We produce about $100,000$ synthetic samples each of size $n$. For each synthetic sample $x$ we compute $T(x)$ using (\ref{Tformula}) provided that the condition $(x-x_0)^t \Sigma_{m}^{\dag}(x-x_0)<\chi^2_{3}(1-\alpha)$ is satisfied. The maximum computed $T$ value is then returned as $T_n(1-\alpha)$.

\subsection{Alternatives}
Two alternative approaches to finite population correction are worth mentioning. Both involve the distribution of synthetic $T$ values just described. The first alternative computes $T$ using (\ref{Tformula}) for each synthetic sample and returns the $1-\alpha$ quantile of the resulting $T$ distribution. The second alternative reports $T(x_0)$ along with its standard error $SE$ which is the standard deviation of synthetic $T$ values. 
\subsection{Does marijuana use cause hard drug use?}
\label{success}
The Population Assessment of Tobacco and Health (PATH) Study in the United States observed a strong association between marijuana use and the use of harder drugs, such as cocaine, methamphetamine, speed, heroin, etc. \citep{PATH}. We obtained data from \citet{PATH} on individual marijuana use, hard drug use, age, and gender (male or female). Detailed descriptions of the utilized variables are provided in Appendix \ref{mar}. Table \ref{Drugs2} shows frequency counts of marijana use, $E$, and hard drug use, $D$.

\begin{table}[ht!]
\centering
\caption{A contingency table showing an observed association (RR $=11.4$, $T=0.58$) between marijuana use and hard drug use.}
\label{Drugs2}
\begin{tabular}{rcc}
\toprule
 & \multicolumn{2}{c}{Marijuana use}\\
\cmidrule{2-3}
Hard drug use&No&Yes\\
\hline
Yes&114&978\\
No&3,649&1,864\\
\hline
\end{tabular}
\end{table}

From the data we compute $P(e=1)=0.43$ and $P(d=1)=0.17$. From \citet{Kendler1998} we estimate $BC_E\approx 0.50$ and from \citet{Kendler1998b} and \citet{T96} we estimate $BC_D\approx 0.40$. Evaluating (\ref{lb1}) with those values produces $l_\eta=0.75$. From the data we compute also $\phi(e,d)=0.42$ to determine the threshold $T=0.58$. Since $\eta\geq l_\eta=0.75>0.58=T$ we have a warrant for causal inference.

However, what if we plan to target our intervention to a subpopulation? Suppose we are interested in conducting causal inference conditional on a vector of covariates. For instance, what if we wanted to infer that marijuana use causes harder drug use on a subpopulation of older males? Conditional data is shown in Table \ref{Drugs}. The conditional association has increased to $RR=14$. From (\ref{Tformula}) we compute the conditional threshold $T=0.50$. Implementing our finite population correction of Section \ref{fpcs} increases that value to $T_{2,000}(95\%)=0.55$.
\begin{table}[ht!]
\centering
\caption{A contingency table showing an observed association (RR $=14$, $T=0.50$, $T_{2,000}(95\%)=0.55$) between marijuana use and the use of hard drugs on a subpopulaton of adult males (age $>35$ years).}
\label{Drugs}
\begin{tabular}{rcc}
\toprule
 & \multicolumn{2}{c}{Marijuana use}\\
\cmidrule{2-3}
Hard drug use&No&Yes\\
\hline
Yes&34&433\\
No&1,015&518\\
\hline
\end{tabular}
\end{table}

\section{Incorporating measured covariate data into the analysis}
\label{scov}
In the previous section we conditioned on a covariate vector to focus our attention on a smaller subpopulation of interest, and then we conducted causal inference on that subpopulation, with a finite population correction. In this section we will remain focused on the original population but incorporate measured covariate data into our analysis. We will show how to compute the threshold $T$ in this more general setting. We introduce a criterion for covariate selection in Section \ref{csel}.

Here we incorporate measured attributes of individuals into a covariate vector $c$ which takes values in a bounded set $C$. 
We denote the vector of covariate measurements with $c$. 
For each individual, covariate measurements should be made prior to exposure and at or prior to the time when $\pi$ is defined. 
We write $m$ for the probability measure corresponding with the population distribution of the vector $c$ over $C$. We write $\mu_c$ for the distribution of $(\pi,r)$ conditional on $c$, and express $\mu$ as a compound distribution: $\mu=\int_{C} \mu_c dm$. The randomness retains its original definition (see \eqref{etadef}).

As before we let $S$ denote the unit square $S=\{(\pi,r):0<\pi,r<1\}$, and we define $\mathcal P_c(S)$ as the class of distributions that produce exact, non-causal, explanations of the observations of $(e,d;c)$. Formally, from the observations of $(e,d;c)$ we know the relative frequencies
\begin{align*}
&p_{01|c}=P(e=0,d=1|c), &&
p_{11|c}=P(e=1,d=1|c), \\
&p_{00|c}=P(e=0,d=0|c),   \quad \textrm{~and~} && 
p_{10|c}=P(e=1,d=0|c),
\end{align*} 
and $\mathcal P_c(S)\subset \mathcal P(S)$ is the class of feasible distributions, satisfying, for each $c$, 
\begin{align*}
&\int(1-\pi)rd\mu(c)=p_{01|c}, &&
 \int \pi rd\mu(c)=p_{11|c}, \\
& \int(1-\pi)(1-r)d\mu(c)=p_{00|c},   \quad \textrm{~and~} && 
\int \pi(1-r)d\mu(c)=p_{10|c}.
\end{align*} 
When $c=\emptyset$ then $\mathcal P_c(S)$ reduces to $\mathcal P(S)$. 
Here we generalize the definition of the threshold $T$ to
\begin{equation}
\label{SRc}
T_c := \max_{\mu\in \mathcal P_c (S)}\eta
\end{equation}
The interpretation is the same as before: if $\eta>T_c$ then causal inference is warranted.

\subsection{Computing the threshold, $T_c$, from observations of $(e,d;c)$}\label{results2}
We write $p_{01|c}=P(e=0,d=1|c)$, $p_{11|c}=P(e=1,d=1|c)$, $p_{00|c}=P(e=0,d=0|c)$, \textrm{and}\quad $p_{10|c}=P(e=1,d=0|c)$, and define 
\[\phi(c) :=\frac{p_{11|c}p_{00|c} - p_{01|c}p_{10|c}}{\sqrt{(p_{11|c}+p_{10|c})(p_{00|c}+p_{10|c})(p_{01|c}+p_{11|c})(p_{01|c}+p_{00|c})}}.
\]
We write $p_{e|c}=P(e=1|c)$ and $p_{d|c}=P(d=1|c)$, and let
\begin{subequations}
\begin{align}
l^2_{\pi(c)}&=p_{d|c}(P(e=1|d=1,c)-p_{e|c})^2+(1-p_{d|c})(P(e=1|d=0,c)-p_{e|c})^2,\\
u^2_{\pi(c)}&=P(e=1|c)(1-P(e=1|c),\\
l^2_{r(c)}&=p_{e|c}(P(d=1|e=1,c)-p_{d|c})^2+(1-p_{e|c})(P(d=1|e=1,c)-p_{d|c})^2, \textrm{~and}\\ 
u^2_{r(c)}&=P(d=1|c)(1-P(d=1|c).
\end{align}
\end{subequations}
With $p_{01}=P(e=0,d=1)$, $p_{11}=P(e=1,d=1)$, $p_{00}=P(e=0,d=0)$, \textrm{and}\quad $p_{10}=P(e=1,d=0)$, define \begin{equation}\label{hatvar}
\sigma^2_{\hat{\pi}(c)}=\int_C ((p_{10|c}+p_{11|c})-(p_{10}+p_{11}))^2dm 
\quad \textrm{and} \quad 
\sigma^2_{\hat{r}(c)}=\int_C ((p_{01|c}+p_{11|c})-(p_{01}+p_{11}))^2dm.
\end{equation}
Write $\sigma^2_{e}$ and $\sigma^2_{d}$ for the measured marginal variances of $e$ and $d$. Write $\sigma^2_{e|c}$ and $\sigma^2_{d|c}$ for the measured variances of $e$ and $d$ conditional on $c$. Write $\sigma^2_{\pi|c}$ and $\sigma^2_{r|c}$ for the variances of $\pi$ and $r$ conditional on $c$. The unmeasured variables $\sigma^2_{\pi|c}$ and $\sigma^2_{r|c}$ are sufficient parameters on $\mathcal{P}_{c}(S)$. The following result isn't a closed form solution, but rather restates the problem as a constrained quadratic optimization problem in terms of conditional variances. 
\begin{theorem}
\label{abcdc}
Suppose $C$ is bounded and for each $c\in C$ that $P(e=1|c)\in (0,1)$ and $P(d=1|c)\in (0,1)$. Define
\begin{subequations}
\label{variables}
\begin{align}
  \label{psi}
  \tau:=\min_{
  \sigma_{\pi|c}^2,  \sigma_{r|c}^2} & \left(\sigma^2_{\hat{\pi}(c)}+\int_C \sigma^2_{\pi|c}dm\right) \ \left(\sigma^2_{\hat{r}(c)}+\int_C \sigma^2_{r|c}dm\right)\\
  \label{variablesB}          
  \textrm{such that } \ &
  \frac{\sigma^2_{\pi|c}\sigma^2_{r|c}}{\sigma^2_{e|c}\sigma^2_{d|c}}= \phi^2(c) 
  &&  \forall c \in C \\
&l^2_{p(c)} \leq \sigma^2_{\pi|c} \leq u^2_{p(c)} 
&&  \forall c \in C\\
&l^2_{r(c)} \leq \sigma^2_{r|c} \leq u^2_{r(c)}
&&  \forall c \in C. 
\end{align}
\end{subequations}
The threshold, $T_c$, of sufficient randomness for causal inference, defined in \eqref{SRc}, can be computed with the following formula: 
\[T_c=1-\sqrt{\tau}/(\sigma_e\sigma_d).\] 
\end{theorem}
\noindent Proposition \ref{abcdc} is proven in Appendix \ref{A3}.

We will often discretize the minimization problem in \eqref{variables} to approximate its solution. If there are few levels of $c$ then we can solve \eqref{variables} with a grid search to determine $\psi$ to arbitrary accuracy. If there are many levels of $c$ then we may employ branch and bound algorithms to solve \eqref{variables} and approximate $\tau$. 
\subsection{Vaccines and Covid-19}
\label{covex}
An early randomized controlled trial of the BNT162b2 mRNA Covid-19 Vaccine demonstrated effectiveness of the vaccine to prevent cases of Covid-19 and also severe cases of Covid-19, but the sample size was not large enough to determine whether or not the vaccine saved lives \citep{Polack2020}. Later and larger observational studies over a set time period have demonstrated that vaccinated individuals are at lower risk of death than unvaccinated individuals \citep{Scobie2021}, but in those later studies vaccination was no longer randomly assigned. Data from a susceptible population is shown in Table \ref{T3}.
\begin{table}[ht]
\centering
\caption{A contingency table showing an observed association (relative risk RR $=0.12$, $T=0.75$) between vaccination and death.}
\label{T3}
\begin{tabular}{rcc}
\toprule
 & \multicolumn{2}{c}{Vaccination}\\
\cmidrule{2-3}
Death&No&Yes\\
\hline
Yes&1,006&188\\
No&6,089&11,102\\
\hline
\end{tabular}
\end{table}

Here the exposure $E$ is vaccination, and the outcome $D$ is death. From the observed data we compute $\phi(e,d)=-0.25$. The threshold is therefore $T=1-|\phi(e,d)|=0.75$. Vaccines for Covid-19 were originally administered to elderly individuals, so perhaps we should adjust our analysis for age. Note that these data were selected, not as part of a scientific study of vaccines and Covid-19, but rather to simply demonstrate certain aspects of our methodology in a setting that is familiar to many readers. The data of Table \ref{tv} were obtained from \citet{Scobie2021} and prepared as described in Appendix \ref{vacc}.

\begin{table}[ht!]
\centering
\caption{A $2\times 2\times 3$ contingency table relating Covid 19 vaccination, death, and age.}
\label{tv}
\begin{tabular}{rccrccrccr}
&  \multicolumn{2}{c}{Age 18--49} && \multicolumn{2}{c}{50--64} && \multicolumn{2}{c}{65+} \\
\toprule
 & \multicolumn{2}{c}{Vaccine} && \multicolumn{2}{c}{Vaccine} && \multicolumn{2}{c}{Vaccine} \\
\cmidrule{2-3} \cmidrule{5-6} \cmidrule{8-9} 
Death&No&Yes&Death&No&Yes&Death&No&Yes\\
\hline
Yes&155&7&Yes&290&23&Yes&561&158\\
No&2,666&1,523&No&1,755&2,447&No&1,668&7,132\\
\hline
\end{tabular}
\end{table}

Here the exposure $e=1$ is vaccination, the outcome $d=1$ is death, and our covariate $c$ is age. The data in Table \ref{tv} give an age-adjusted relative risk of adjRR $=0.08$. We applied Theorem \ref{abcdc} to the data of Table \ref{tv} to compute $T_c=0.70$. Adjustment for age lowered the threshold from $T=0.75$ to $T_c=0.70$. In general, incorporation of more measured covariates into the analysis can only lower $T$ further since the optimization problem in (\ref{SRc}) becomes more constrained.

\subsection{Covariate Selection}
\label{csel}
There is some controversy regarding covariate selection. Is it appropriate to adjust for all pre-treatment or pre-exposure covariates? \citet{DM} have conducted simulations to investigate M bias and butterfly bias. More relevant to our proposed methodology is the issue of z-bias \citep{Wool,DR,Pim}. 

The idea behind the issue of z-bias is the following. If there are irrelevant, pre-exposure variables we might not want to condition on them because it is desirable for irrelevant variables to determine exposure assignment. Conditioning on those irrelevant variables could increase the proportion of residual, exposure variation due to unmeasured confounders. Utilizing the language of our proposed methodology, we can analogously formulate the ample randomness criterion. 

\begin{definition}[The ample randomness criterion]
\label{ample}
Covariates should be selected to increase $\eta/T$.
\end{definition}

Analysis of Theorem \ref{abcdc} suggests that lower $T$ values are possible when $s(c):=(P(e=1|c),P(d=1|c))$ varies over the unit square, $S=(0,1)^2$, and that fact can guide covariate selection. Note, given $c$, that $P(e=1|c)$ is a propensity score and $P(d=1|c)$ is a prognostic score. A desire for variation in $s(c)$ is consistent with the common practice of selecting pre-exposure covariates that are strong predictors of $e$ and $d$.

However, caution is advised when conditioning on a covariate that was itself randomly assigned, as that act of conditioning may eliminate part of the randomness $\eta$. We recommend conditioning on fixed characteristics of individuals such as genomic variables, birth date (or age), and birthplace, rather than chance events that happened to occur. However, in some situations it may be reasonable to condition on chance events to reduce confounding \citep{VanderWeele2011}.
\section{Discussion}
\label{disc}
There are numerous methods for sensitivity analysis, including but not limited to \citet{Frank,HHH,SAWA,EV,KD,Oster,CH,KOA}. As evident in \citet{Ding2014}, many of those methods can be traced back to the seminal work of \citet{Cornfield}. In \citet{Cornfield} the following is written:
\begin{quote}The magnitude of the excess lung-cancer risk among cigarette smokers is so great that the results can not be interpreted as arising from an indirect association of cigarette smoking with some other agent or characteristic, since this hypothetical agent would have to be at least as strongly associated with lung cancer as cigarette use; no such agent has been found or suggested.
\end{quote}
Their argument is reasonable, but not watertight. Absence of evidence is not evidence of absence; see \cite{FERES2023}. The introduced methodology of this paper utilizes a sensitivity parameter $T$, but the introduced methodology becomes something more than a sensitivity analysis when paired with knowledge of the randomness $\eta$. The condition $T<\eta$ is evidence for the absence of confounding, assuming SUTVA.

The main limitation of our proposed methodology is its reliance on SUTVA. Here we are concerned not with the no-interference component of SUTVA but rather we are concerned that the potential outcomes are not well defined. We have interpreted the discordance of monozygotic twin studies as evidence for randomness within processes that give rise to exposures and outcomes, but that natural randomness could assign multiple versions of treatment. 

Our methodology conclusively establishes the existence of a cause after the time at which propensity and prognosis probabilities are defined, but further scientific investigation, e.g. the negative controls of \cite{Bjornevik2022}, will likely be needed. Our proposed methodology is not meant to stand alone but designed to play a supporting role; see Section \ref{sens}, Definition \ref{ample}, and scientifically oriented discussion in \cite{Rosenbaum2015} and \cite{Pearce}. 

Our formulas and algorithms for computing $T$, $T_n(1-\alpha)$, and $T_c$ are efficient as long as the covariate vector is low dimensional, but it would be time consuming to compute $T_c$ with a high dimensional vector of covariates and a finite population correction. One way to overcome this curse of dimensionality is to condition on the joint-probabilities of $(e,d)$ instead of $c$; c.f. \cite{Rosenbaum1983}.

Twin studies are known to provide value \citep{Hagenbeek2023}, and there are many twin registries and ongoing longitudinal twin studies providing opportunities to analyze data from exposure discordant twins and subsequent disease \citep{Sahu2016}. There is reason to believe that past populations of twins were representative of the general population \citep{Lykken1982}, but more recent populations of monozygotic twins may be less representative of the general population due to modern assisted reproductive technology \citep{Vitthala2008}. Also, it's possible that modern behavior is less random; c.f. \citet{VanderWeele2010}. 

The introduced methodology of this paper is flexible enough to allow an upper bound of $R^2_\pi$ to be estimated from one study while an upper bound of $R^2_r$ is estimated from another study, provided valid transport. A basic requirement for valid transport is that the observed proportion $P(e=1)$ should match the proportion exposed in the twin study, and likewise for the outcome; see \citet[Figure 2b]{Hagenbeek2023}. For brevity we have included numerous simple examples of transport here, but in practice we recommend more thorough and careful literature reviews to ensure valid transport from representative populations. Multiple twin studies can be analyzed, and attention can be paid to sampling error, ascertainment bias, and other issues as needed.  

Transportability is an emerging area of interest in causal inference \citep{Mitra2022}. Here we have shown how transport of measures of concordance in monozygotic twin studies can inform causal inference. It is important to note that our proposed methodology does not require the use of twin studies. Any approach that identifies randomness in the data generating process will suffice. We have used twin studies here as a simple way to demonstrate how to estimate the randomness $\eta$ empirically. We have shown a proof of concept that goes beyond standard approaches to sensitivity analysis. We have shown that observational causality testing is possible.



\appendix
\section{Concordance formulas}
\label{concordances}
 Let $n$ be the number of twin pairs. Let $C$ be the count of pairs where both have the trait. Let $\tilde{D}$ be the count of twin pairs where one but not both have the trait. Let $U=n-C-\tilde{D}$ be the number of pairs where neither have the trait. Denote the probandwise concordance with $BC=2C/(2C+\tilde{D})=C/(C+\tilde{D}/2)$. Denote the pairwise concordance with $PC=C/(C+\tilde{D})$. Define $V=(U+C)/n$. $V$ is a less practical but more mathematical way of defining concordance. For individual $i$ let $\psi_i$ be their propensity probability of the trait. The proportion of those with the trait is $\bar{\psi}=(C+\tilde{D}/2)/n$. Note that when the trait is the exposure $E$ then $\bar{\psi}=P(e=1)$ and when the trait is the outcome $D$ then $\bar{\psi}=P(d=1)$. Due to common causes we assume $L:=\sum_{i=1}^n(\psi_i^2+(1-\psi_i)^2)/n\leq V$. The inequality will be tighter if the counts are obtained from studies of monozygotic twins reared apart, but there is still the common womb environment to consider. Define $\sigma^2_\psi = \sum_{i=1}^n (\psi_i-\bar{\psi})^2/n$ and $R^2=\sigma^2/(\bar{\psi}(1-\bar{\psi})$.
\begin{lemma}
\label{probandwiselemma}
    $V=1-2\bar{\psi}(1-BC) = 1-2\bar{\psi} \frac{1-PC}{1+PC}$
\end{lemma}
\begin{proof} We compute
    \begin{align*}
1- 2\bar{\psi}\left(1-BC\right)=1+ 2\bar{\psi}\left(BC-1\right) &= 1+2\left(\frac{C+\tilde{D}/2}{n}\right)\left(\frac{C}{C+\tilde{D}/2}-1\right)\\
&= 1+2\left(\frac{C+\tilde{D}/2}{n}\right)\left(\frac{-\tilde{D}/2}{C+\tilde{D}/2}\right)\\
&= 1-\frac{\tilde{D}}{n}= \frac{n-\tilde{D}}{n}= \frac{U+C}{n}=V.
\end{align*}
Also, 
\begin{equation}
\label{e:prob2pair}
\frac{1-PC}{1+PC}=\frac{(\tilde{D}/(C+\tilde{D}))}{(2C+\tilde{D})/(C+\tilde{D})}=\frac{\tilde{D}}{2C+\tilde{D}}=1-BC.
\end{equation}
\end{proof}
\begin{lemma}
\label{sigmalemma}
    $R^2=1 -\frac{\left(1-L \right)}{2\bar{\psi}(1-\bar{\psi})}$
\end{lemma}
\begin{proof}
\begin{align*}
L  = \frac{1}{n}\sum_{i=1}^n \left(\psi_i^2+ \left(1-\psi_i \right)^2 \right)
&= \frac{1}{n}\sum_{i=1}^n \left(\psi_i^2+1 -2\psi_i + \psi_i^2\right) 
= \frac{2}{n}\sum_{i=1}^n \psi_i^2 +1 -2\bar{\psi}.\\
\end{align*}
Thus, $\frac{L-1}{2}= \sum_{i=1}^n\psi_i^2/n-\bar{\psi}$.
The variance formula then gives
\begin{align*}
    \sigma^2 
    & = \sum_{i=1}^n \psi_i^2/n - \bar{\psi}^2 
     = \frac{L-1}{2}+ \bar{\psi} - \bar{\psi}^2
     = \bar{\psi}(1-\bar{\psi}) -\frac{\left(1-L \right)}{2}.
\end{align*}
It remains only to divide $\sigma^2$ by $\bar{\psi}(1-\bar{\psi})$ and recall \eqref{e:prob2pair}.
\end{proof}

\begin{proposition}
\label{probandwiseprop}
$R^2 \leq 1-\frac{1}{1-\bar{\psi}}(1-BC) = 1-\frac{1}{1-\bar{\psi}}\frac{1-PC}{1+PC}$
\end{proposition}
\begin{proof}
By Lemma \ref{sigmalemma} and the assumption that $L\leq V$ we have \begin{equation}\label{proofeq}\sigma^2=\bar{\psi}(1-\bar{\psi})+\frac{L-1}{2}\leq \bar{\psi}(1-\bar{\psi})+\frac{V-1}{2}.\end{equation} By Lemma \ref{probandwiselemma} we have $V=1-2\bar{\pi}(1-BC)$ implying $(V-1)/2=\bar{\psi}(1-BC)$. It remains only to divide (\ref{proofeq}) by $\bar{\psi}(1-\bar{\psi})$ and recall \eqref{e:prob2pair}.
\end{proof}

\section{Proof of Proposition \ref{P}}
\label{A1}
For each measure of association, $\xi \in \{RD,RR,OR\}$, we derive the formula in \eqref{RDformula}. 

\medskip

\noindent \underline{$\xi = RD$.}
By definition, 
$RD=\frac{p_{11}}{p_{11}+p_{10}}-\frac{p_{01}}{p_{00}+p_{01}}$ so 
\begin{align*}
p_e(1-p_e)RD
= (p_{00}+p_{01}) p_{11} - (p_{11}+p_{10})   p_{01} 
= p_{00}p_{11} - p_{10} p_{01} 
= \phi \sqrt{p_e (1-p_e) p_d (1-p_d)}
\end{align*}
So we have that $\lambda |RD|=  |\phi| $, so $T = 1 - |\phi| = 1 - |RD| \lambda$, as desired.

\bigskip

\noindent \underline{$\xi = RR$.} 
Using the definition of $RR$, we have 
\begin{align*}
RR - 1 &= \frac{p_{11}}{p_e} - \frac{1-p_e}{p_{01}} - 1 
= 
\frac{p_{11} - p_e p_d}{p_e p_{01}} \\ 
p_d(RR-1) &= \frac{p_d}{p_e}\frac{p_{11} - p_e p_d}{p_{01}}\\
1 + p_e (RR-1) & = \frac{p_d(1-p_e)}{p_{01}}.
\end{align*}
Dividing these last two quantities, we obtain
$$
\frac{p_d(RR-1)}{1 + p_e (RR-1)} = 
\frac{p_{11} - p_e p_d}{p_e(1-p_e)}. 
$$
On the other hand, 
$$RD = \frac{p_{11}}{p_e} - \frac{p_{01}}{1-p_e} = \frac{p_{11}(1-p_e) - p_{01}p_e}{p_e(1-p_e)} 
= \frac{p_{11} - p_e p_d}{p_e(1-p_e)}.
$$
Combining these last two expressions, we have that 
\begin{equation*}
RD=\frac{ p_d(RR-1)}{1+p_e(RR-1)},
\end{equation*} 
which together with the expression  $T = 1 - |RD|\lambda$ gives the desired result.

\bigskip

\noindent \underline{$\xi = OR$.} 
We compute 
\begin{align}
\label{OR/RR}
    \frac{OR}{RR} 
       & = \frac{p_{00}}{p_{10}}\frac{(p_{10}+p_{11})}{(p_{01}+p_{00})}\\
       \nonumber
       & = \frac{p_{10}p_{00}+ p_{10}p_{01} - p_{10}p_{01}+p_{00}p_{11}}{p_{10}p_{01}+p_{10}p_{00}}\\
       \nonumber
       & = 1-\frac{p_{10}p_{01}}{p_{10}p_{00}+p_{10}p_{01}} + \frac{p_{00}}{p_{10}} \left(\frac{p_{11}}{p_{00}+p_{01}} \right)\\ 
       \nonumber
       & = 1-\frac{p_{01}}{p_{00}+p_{01}} + \frac{p_{01}}{p_{00} + p_{01}}OR. 
\end{align}
Rearranging, we obtain
\begin{equation}
\label{RROR}
RR=\frac{OR}{1-\frac{p_{01}}{p_{01}+p_{00}}+\frac{p_{01}}{p_{01}+p_{00}}OR}.
\end{equation}
We proceed to express $\frac{p_{01}}{p_{01}+p_{00}}$ in terms of $RR$, $p_e$, and $p_d$. By the definition of $RR$ we have
\begin{equation}\label{4comb}RR=\frac{p_{11}}{p_{01}}\frac{(1-p_e)}{p_e}.\end{equation}
From \eqref{4comb} and $p_d=p_{01}+p_{11}$ we solve for 
\[p_{01}=\frac{p_d(1-p_e)}{1+p_e(RR-1)}\]
and then divide by $p_{01}+p_{00}=1-p_e$ to obtain
\begin{equation}\label{4ac}\frac{p_{01}}{p_{01}+p_{00}}=\frac{p_d}{1+p_e(RR-1)}.\end{equation}
By combining \eqref{RROR} and \eqref{4ac} we obtain
\begin{equation}
\label{quadrat}
p_e RR^2+(p_d(OR-1)+(1-p_e)-p_eOR)RR-(1-p_e)OR=0.
\end{equation}
Writing  $a = p_d\left(OR-1\right)+\left(1-p_e\right)-p_eOR$ and $m = (p_e-1)OR$, we apply the quadratic formula to find that 
$$
RR= u_\pm := \frac{-a \pm \sqrt{a^2-4 p_e m}}{2 p_e}.
$$
The discriminant is clearly positive since 
$a^2-4 p_e m = a^2 + 4 p_e(1-p_e)OR \geq 0$ so we have two real roots. 
To choose the sign, we insist that $RR=1$ when $OR=1$, which requires the plus sign in the quadratic formula. Now combining the result with the formula for $T(\xi)$ with $\xi = RR$ in \eqref{RDformula}, we obtain the desired formula.
\begin{flushright}
$\square$
\end{flushright}
\section{Data preparation}
\label{dataprep}
\subsection{Marijuana and hard drug use}
\label{mar}
We obtained data from The Population Assessment of Tobacco and Health (PATH) Study \cite{PATH}. From ICPSR 36498 we downloaded DS1001 Wave 1: Adult Questionnaire Data with Weights. We utilized raw (unweighted) data on six variables. The first variable was (our exposure $e$) R01\_AX0085 
(Ever used marijuana, hash, THC or grass). Respondents answered yes or no. The outcome variable (our disease $d$) was derived from three variables: R01\_AX0220\_01 (Ever used substance: Cocaine or crack), R01\_AX0220\_02 (Ever used substance: Stimulants like methamphetamine or speed), and R01\_AX0220\_03 (Ever used substance: Any other drugs like heroin, inhalants, solvents, or hallucinogens). We recorded $d=1$ if a respondent answered yes to any of the questions R01\_AX0220\_01, R01\_AX0220\_02, or R01\_AX0220\_03, and we recorded $d=0$ otherwise. The two measured covariates were the following:
R01R\_A\_AGECAT7 (Age range when interviewed (7 levels)) and R01R\_A\_SEX (Gender from the interview (male or female)). With regards to age individuals were broadly classified as being $\geq 35$ years of age or $< 35$ years of age.
\subsection{Vaccines and Covid-19}
\label{vacc}
We obtained raw data from the table published in \cite{Scobie2021}. We accessed raw frequency data from the ``June 20–July 17'' section of that table. We assumed that the vaccine was $90\%$ effective at preventing severe cases of Covid-19. In the analysis of \cite{Scobie2021} they assumed vaccine effectiveness of $80\%$, $90\%$, and $95\%$. Our analysis used their median assumption of $90\%$. We conducted inverse probability weighting of the raw frequencies in the hospitalized and vaccinated column of the table of \cite{Scobie2021}. We assumed that any vaccinated person who died from Covid-19 was also hospitalized. If amongst the hospitalized and vaccinated of a given age group the raw frequencies were $a$ deaths and $b$ survivors, we replaced $b$ with $(a+b)*10-a$. We did this for all age groups to obtain the data of \ref{tv}. The data of \ref{T3} is the resulting marginal table ignoring age.
\section{Proof of Proposition \ref{abcdc}}
\label{A3}
In (\ref{SRc}) the objective function to maximize is \[\eta=1-\frac{\sigma_\pi\sigma_r}{\sqrt{\bar{\pi}(1-\bar{\pi})\bar{r}(1-\bar{r})}}.\]
The denominator $\sqrt{\bar{\pi}(1-\bar{\pi})\bar{r}(1-\bar{r})}$ is fixed by the constraints. Since the function $x\mapsto \sqrt{x}$ on $x>0$ is monotonically increasing, it is equivalent  to minimize $f(\mu):=\sigma^2_{\pi}(\mu)\sigma^2_r(\mu)$. By the law of total variation, $$
f(\mu)=\left( \sigma^2_{\hat{\pi}(c)} +\int_C \sigma^2_{\pi|c}dm\right) \ \left( \sigma^2_{\hat{r}(c)} +\int_C \sigma^2_{r|c}dm\right), 
$$ 
which is the objective function appearing in \eqref{variables}. 

We next consider the constraints in  \eqref{variables}. Note that for any distribution component $\mu_c \in \mathcal{P}(S)$, we have
$$
\sigma^2_{\pi|c} 
= \int \pi^2d\mu_c - \left(\int \pi d\mu_c\right)^2 
\leq \int \pi d\mu_c - \left(\int \pi d\mu_c\right)^2 
= u_{\pi(c)}^2
$$
and similarly $\sigma^2_{r|c} \leq u^2_{r(c)}$, 
which are upper bound constraints in \eqref{variables}. 
We also observe that for any distribution component $\mu_c \in \mathcal{P}(S)$,  we can write 
$$
\mu_c = (1-p_{e|c} )\mu_{c,0} +p_{e|c} \mu_{c,1},
$$
where $\mu_{c,0}$ is a distribution of points with $e=0$, and $\mu_{c,1}$ is a distribution of points with $e=1$. Using the law of total variation, we then have  \begin{align*}
\sigma^2_{\pi|c} 
&= (1-p_{e|c})\left(\int \pi d\mu_{c,0}-\int \pi d\mu_c \right)^2 \\
&+p_{e|c} \left(\int \pi d\mu_{c,1}-\int \pi d\mu_c \right)^2\\
&+\sigma^2_p(\mu_{c,0})+\sigma^2_p(\mu_{c,1}) \\ 
& \geq (1-p_{e|c})\left(\int \pi d\mu_{c,0}-\int \pi d\mu_c \right)^2 \\
&+ p_{e|c} \left(\int \pi d\mu_{c,1}-\int \pi d\mu_c \right)^2 \\
& = l^2_{p(c)}.
\end{align*}
Similarly, $\sigma^2_{r|c}  \geq l^2_{r(c)}$.
This shows that the lower bounds in  \eqref{variables} hold. 
By \citet[Theorem 1]{Knaeble2023}, we have  
$\frac{\sigma^2_{\pi|c}\sigma^2_{r|c}}{\sigma^2_{e|c}\sigma^2_{d|c}}\geq \phi^2(c)$, $\forall c \in C$. 
Since $f(\mu)$ is non-increasing in $\sigma^2_{\pi|c}$ and $\sigma^2_{r|c}$, we can consider only compound distributions $\mu=\int_{C} \mu_c dm$ satisfying 
$$
 \frac{\sigma^2_{\pi|c}\sigma^2_{r|c}}{\sigma^2_{e|c}\sigma^2_{d|c}} = \phi^2(c), 
  \qquad \qquad  \forall c \in C, 
$$
as in \eqref{variablesB}. 

This motivates the form of the optimization problem in \eqref{variables}. Provided $C$ is bounded, a solution to \eqref{variables}  exists because the objective function is continuous and the constraint set is compact with respect to the weak-$\star$ topology. Finally, we observe that for any collection
$\{\sigma^2_{\pi|c}, \sigma^2_{r|c}\}_{c \in C}$ 
that satisfies the constraints in \eqref{variables}, there exists a compound distribution $\mu=\int_{C} \mu_c dm$ satisfying the constraints of (\ref{SRc}) with variances $\sigma^2_{\pi|c}$ and $\sigma^2_{r|c}$. Thus, the solution to \eqref{variables} yields the optimal value in (\ref{SRc}). 
\section{Conservative confidence sets}
\label{appmatrix}
This section of the appendix is concerned with the $(1-\alpha)\times 100\%$ confidence set $\mathcal P_{n,1-\alpha}(S)$ defined in (\ref{finiteform}). That definition is based on a central limit theorem approximation provided $n\gg 3$; see \cite[p. 176]{JW}. The confidence set is conservative because it was constructed with a conservative, multinomial approximation of the Generalized (Poisson) Binomial distribution. The Generalized (Poisson) Binomial distribution (see \cite{Beaulieu1991}) is the true distribution of the frequency counts given fixed $\mu\in \mathcal P_n(S)$. The following proposition shows how the multinomial approximation is conservative. Define the discrete simplex 
\begin{equation} 
\label{e:DeltaN}
\Delta_n=\{(x_{01},x_{11},x_{00},x_{10})\in \mathbb{N}_0^4 \colon  x_{01}+x_{11}+x_{00}+x_{10}=n\}.
\end{equation}
Given fixed $\mu\in \mathcal P_n(S)$ let $X(\mu)$ be the random variable of frequency counts on $\Delta_n$ distributed according to the Generalized (Poisson) Binomial distribution, and let $X(\bar{p})$ be the random variable of frequency counts on $\Delta_n$ distributed according to the multinomial distribution of $n$ trials with probabilities determined from the same $\mu$ by $\bar{p}=f(\mu)/n$, where $f$ is defined in (\ref{vfunctional}). Note that $E(X(\mu))=E(X(\bar{p}))$.
\begin{proposition}
\label{conservativecs}
Let $\Sigma_{X(\mu)}$ denote the covariance matrix of $X(\mu)$, and let $\Sigma_{X(\bar{p})}$ denote the covariance matrix of $X(\bar{p})$. We have \[\Sigma_{X(\bar{p})} \succeq \Sigma_{X(\mu)} ,\]
where the Loewner semidefinite ordering $A \succeq B $ means that $A-B$ is a positive semi-definite matrix. 
\end{proposition}
We write $$X(\mu)=(X_{01}(\mu), X_{11}(\mu), X_{00}(\mu), X_{10}(\mu))$$ and $$X(\bar{p})=(X_{01}(\bar{p}), X_{11}(\bar{p}), X_{00}(\bar{p}), X_{10}(\bar{p})).$$ We compute the symmetric covariance matrix
\begin{equation*}
\Sigma_{X(\mu)}= 
\begin{bmatrix}
\sigma^2_{X_{01}(\mu)} & \sigma_{X_{01}(\mu)X_{11}(\mu)} & \sigma_{X_{01}(\mu)X_{00}(\mu)} & \sigma_{X_{01}(\mu)X_{10}(\mu)}\\
\cdot & \sigma^2_{X_{11}(\mu)} & \sigma_{X_{11}(\mu)X_{00}(\mu)} & \sigma_{X_{11}(\mu)X_{10}(\mu)}\\
\cdot & \cdot & \sigma^2_{X_{00}(\mu)} & \sigma_{X_{00}(\mu)X_{10}(\mu)}\\
\cdot & \cdot & \cdot & \sigma^2_{X_{11}(\mu)}
\end{bmatrix},
\end{equation*}
where
\begin{align*}
    &\sigma^2_{X_{01}(\mu)}=\sum_{i=1}^n (1-\pi_i)r_i(1-(1-\pi_i)r_i), \quad \sigma^2_{X_{11}(\mu)}=\sum_{i=1}^n \pi_i r_i(1-\pi_i r_i), \\
    &\sigma^2_{X_{00}(\mu)}=\sum_{i=1}^n (1-\pi_i )(1-r_i)(1-(1-\pi_i)(1-r_i)), \quad \sigma^2_{X_{10}(\mu)}=\sum_{i=1}^n \pi_i(1-r_i)(1-\pi_i(1-r_i)),\\
    &\sigma_{X_{01}(\mu)X_{11}(\mu)}=-\sum_{i=1}^n \pi_i(1-\pi_i)r^2_i, \quad \sigma_{X_{01}(\mu)X_{00}(\mu)}=-\sum_{i=1}^n (1-\pi_i)^2r_i(1-r_i), \\
    &\sigma_{X_{11}(\mu)X_{10}(\mu)}=-\sum_{i=1}^n p^2_ir_i(1-r_i), \textrm{~and~} \\
&\sigma_{X_{01}(\mu)X_{10}(\mu)}=\sigma_{X_{11}(\mu)X_{00}(\mu)}=-\sum_{i=1}^n \pi_i(1-\pi_i)r_i(1-r_i).
\end{align*}
We then compute the other symmetric covariance matrix
\begin{equation*}
\Sigma_{X(\bar{p})}= 
\begin{bmatrix}
\sigma^2_{X_{01}(\bar{p})} & \sigma_{X_{01}(\bar{p})X_{11}(\bar{p})} & \sigma_{X_{01}(\bar{p})X_{00}(\bar{p})} & \sigma_{X_{01}(\bar{p})X_{10}(\bar{p})}\\
\cdot & \sigma^2_{X_{11}(\bar{p})} & \sigma_{X_{11}(\bar{p})X_{00}(\bar{p})} & \sigma_{X_{11}(\bar{p})X_{10}(\bar{p})}\\
\cdot & \cdot & \sigma^2_{X_{00}(\bar{p})} & \sigma_{X_{00}(\bar{p})X_{10}(\bar{p})}\\
\cdot & \cdot & \cdot & \sigma^2_{X_{11}(\bar{p})}
\end{bmatrix},
\end{equation*}
where
\begin{align*}
    &\sigma^2_{X_{01}(\bar{p})}=\frac{1}{n}\sum_{i=1}^n (1-\pi_i)r_i\sum_{i=1}^n(1-(1-\pi_i)r_i), \quad \sigma^2_{X_{11}(\bar{p})}=\frac{1}{n}\sum_{i=1}^n \pi_i r_i\sum_{i=1}^n(1-\pi_ir_i), \\
    &\sigma^2_{X_{00}(\bar{p})}=\frac{1}{n}\sum_{i=1}^n (1-\pi_i)(1-r_i)\sum_{i=1}^n(1-(1-\pi_i)(1-r_i)),\\
& \frac{1}{n}\sigma^2_{X_{10}(\bar{p})}=\sum_{i=1}^n \pi_i(1-r_i)\sum_{i=1}^n(1-\pi_i(1-r_i)),\\
    &\sigma_{X_{01}(\bar{p})X_{11}(\bar{p})}=\frac{-1}{n}\sum_{i=1}^n (1-\pi_i)r_i\sum_{i=1}^n\pi_ir_i, \\
 &\sigma_{X_{01}(\bar{p})X_{00}(\bar{p})}=\frac{-1}{n}\sum_{i=1}^n (1-\pi_i)r_i\sum_{i=1}^n(1-\pi_i)(1-r_i), \\
    &\sigma_{X_{11}(\bar{p})X_{10}(\bar{p})}=\frac{-1}{n}\sum_{i=1}^n \pi_ir_i\sum_{i=1}^n\pi_i(1-r_i), \textrm{~and~}\\ &\sigma_{X_{01}(\bar{p})X_{10}(\bar{p})}=\sigma_{X_{11}(\bar{p})X_{00}(\bar{p})}=\frac{-1}{n}\sum_{i=1}^n \pi_i(1-r_i)\sum_{i=1}^nr_i(1-\pi_i).
\end{align*}

We will show that $\Sigma_{X(\mu)} - \Sigma_{X(\bar{p})}$ is diagonally dominant. By symmetry of the unit square it suffices to consider a single row.  We describe the calculations for the second row. In what follows the expectations are over the population. It suffices to show
 \begin{align*}
 \left| E_i\pi_ir_i(1-\pi_ir_i) - E_i\pi_ir_iE_i(1-\pi_ir_i) \right| 
 \ge \left|E_i(1- \pi_i)r_i\pi_ir_i - E_i(1-\pi_i)r_iE_i\pi_ir_i\right|  \\
     +  \left| E_i\pi_ir_i(1-\pi_i)(1-r_i) - E_i\pi_ir_iE_i(1-\pi_i)(1-r_i) \right| \\
   + \left| E_i\pi_i r_i \pi_i(1-r_i)- E_i\pi_ir_i E_i \pi_i(1-r_i)\right|.
  \end{align*}
 Repeated applications of Jensen's inequality allows us to drop the absolute value bars. We then need to show
   \begin{align*}
  E_i\pi_ir_i(1-\pi_ir_i) - E_i\pi_ir_iE_i(1-\pi_ir_i)  
 - E_i(1- \pi_i)r_i\pi_ir_i + E_i(1-\pi_i)r_iE_i\pi_ir_i  \\
     -   E_i\pi_ir_i(1-\pi_i)(1-r_i) + E_i\pi_ir_iE_i(1-\pi_i)(1-r_i)  \\
   -  E_i\pi_i r_i \pi_i(1-r_i) +  E_i \pi_i r_i E_i \pi_i(1-r_i) \ge 0.
  \end{align*}
Rearranging, we have
     \begin{align*}
  \left(E_i\pi_ir_i(1-\pi_ir_i) - E_i(1- \pi_i)r_i\pi_ir_i  - E_i\pi_ir_i(1-\pi_i)(1-r_i)  -  E_i \pi_i r_i \pi_i(1-r_i)\right) -\\ \left(E_i\pi_ir_iE_i(1-\pi_ir_i)  
 - E_i(1-\pi_i)r_iE_i\pi_ir_i  - E_i\pi_ir_iE_i(1-\pi_i)(1-r_i)  - E_i \pi_ir_i E_i \pi_i(1-r_i) \right)=0,
  \end{align*}
which is greater than or equal to zero as desired.
\begin{flushright}
$\square$
\end{flushright}
\clearpage
\bibliographystyle{abbrvnat}
\bibliography{refs.bib}
\end{document}